\documentclass[twocolumn]{article}
\pdfoutput=1
\usepackage[utf8]{inputenc}
\usepackage{amsmath}
\usepackage{amssymb}
\usepackage{array}
\usepackage{clrscode3e}
\usepackage{float}
\usepackage{url}
\usepackage{xspace}
\usepackage{graphicx}
\usepackage{enumerate}

\makeatletter
\newcommand{\defineabbrev}[1]{\@namedef{#1}{\textsc{\lowercase{#1}}\xspace}}
\makeatother

\defineabbrev{GCD}
\defineabbrev{GMP}
\defineabbrev{HGCD}

\newcommand{\floor}[1]{\lfloor #1 \rfloor}

\newtheorem{proposition}{Proposition}
\newenvironment{proof}
{\par\noindent\textbf{Proof}:}
{\hfill\quad$\square$\par\smallskip}

\newfloat{algorithm}{tp}{loa}
\floatname{algorithm}{Algorithm}

\newcommand{\Z}{\mathbb{Z}}
\renewcommand{\gets}{\leftarrow}

\title{Efficient computation of the Jacobi symbol}
\author{Niels Möller}
\date{2019}

\begin{document}
\maketitle

\begin{abstract}
  The family of left-to-right \GCD algorithms reduces input numbers by
  repeatedly subtracting the smaller number, or multiple of the
  smaller number, from the larger number. This paper describes how to
  extend any such algorithm to compute the Jacobi symbol, using a
  single table lookup per reduction. For both quadratic time \GCD
  algorithms (Euclid, Lehmer) and subquadratic algorithms (Knuth,
  Schönhage, Möller), the additional cost is linear, roughly one table
  lookup per quotient in the quotient sequence. This method was used
  for the 2010 rewrite of the Jacobi symbol computation in \GMP.
\end{abstract}

\section{Introduction}

The Legendre symbol and its generalizations, the Jacobi symbol and the
Kronecker symbol, are important functions in number theory. For
simplicity, in this paper we focus on computation of the Jacobi
symbol, since the Kronecker symbol can be computed by the same
function with a little preprocessing of the inputs.

\subsection{Jacobi and GCD}

Two quadratic algorithms for computing the Kronecker symbol
(and hence also the Jacobi symbol) are described as Algorithm~1.4.10
and~1.4.12 in~\cite{cohen}. These algorithms run in quadratic time,
and consists of a series of reduction steps, related to Euclid's \GCD
algorithm and the binary \GCD algorithm, respectively. Both Kronecker
algorithms share one property with the binary \GCD algorithm: The
reduction steps examine the current pair of numbers in both ends. They
examine the least significant end to cast out powers of two, and they
examine the most significant end to determine a quotient (like in
Euclid's algorithm) or to determine which number is largest (like in
the binary \GCD algorithm).

Fast, subquadratic, \GCD algorithms work by di\-vide-and-conquer, where
a substantial piece of the work is done by examining only one half of
the input numbers. Fast left-to-right \GCD is related to fast
algorithms for computing the continued fraction
expansion~\cite{schoenhage:1971,knuth:algorithms}. These are
left-to-right algorithms, in that they process the input from the most
significant end. The binary recursive algorithm~\cite{stehle} is a
right-to-left algorithm, in that it processes inputs from the least
significant end. The asymptotic running times of these algorithms are
$O(M(n) \log n)$, where $M(n)$ denotes the time needed to multiply two
$n$-bit numbers. The \GCD algorithm used in recent versions of
the \GMP library~\cite{gmp} is a variant of Schönhage's
algorithm~\cite{moller-sgcd}.

It is possible to compute the Jacobi symbol in subquadratic time,
with the same asymptotic complexity as \GCD. One algorithm is
described in~\cite{bach-shallit} (solution to exercise~5.52), which
says:
\begin{quotation}
  This complexity bound is part of the ``folklore'' and has
  apparently never appeared in print. The basic idea can be found in
  Gauss [1876]. Our presentation is based on that in Bachmann [1902].
  H. W. Lenstra, Jr. also informed us of this idea; he attributes it
  to A. Schönhage.
\end{quotation}

Since the quadratic algorithms for the Jacobi symbol examines the data
at both ends, some reorganization is necessary to construct a
divide-and-conquer algorithm that processes data from one end. The
binary \GCD algorithm has the same problem. In the binary recursive
\GCD algorithm, this is handled by using a slightly different
reduction step using 2-adic division.

Recently, the binary recursive \GCD algorithm has been extended to
compute the Jacobi symbol~\cite{brent:jacobi}. The main difference to
the corresponding \GCD algorithm is that it needs the intermediate
reduced values to be non-negative, and to ensure this the binary
quotients must be chosen in the range $1 \leq q < 2^{k+1}$ rather than
$|q| < 2^k$. As a result, the algorithm is slower than the
\GCD algorithm by a small constant factor.

\subsection{Main contribution}

This paper describes a fairly simple extension to a
wide class of left-to-right \GCD algorithms, including Lehmer's
algorithm and the subquadratic algorithm in~\cite{moller-sgcd}, which
computes the Jacobi symbol using only $O(n)$ extra time and $O(1)$
extra space\footnote{The size of the additional state to be maintained
  is $O(1)$. But in a practical implementation, which does not store
  this state in a global variable, either the state or a pointer to it
  will be copied into each activation record, which for a subquadratic
  recursive divide-and-conquer algorithm costs $O(\log n)$ extra
  space rather than $O(1)$}. This indicates that also for the fastest
algorithms for large inputs, the cost is essentially the same for
computing the \GCD and computing the Jacobi symbol.\footnote{Even though
  we cannot rule out the existence of a left-to-right \GCD algorithm
  which is a constant factor faster than Jacobi. Such an algorithm
  would lie outside the class of ``generic left-to-right \GCD
  algorithms'' we describe in this paper, e.g., it might use
  intermediate reduced values of varying signs and quotients that are
  rounded towards the nearest integer rather than towards $-\infty$.}

Like the algorithm described in~\cite{bach-shallit}, the computation
is related to the quotient sequence. The updates of the Jacobi symbol
are somewhat different, instead following an unpublished algorithm by
Schönhage~\cite{schoenhage-brent-communication} for computing the
Jacobi symbol from the quotient sequence modulo four. In the \GCD
algorithms in \GMP, the quotients are not always applied in a single
step; instead, there is a series of reductions of the form $a \gets a
- m b$, where $m$ is a positive number equal to or less than the
correct quotient $\floor{a/b}$. In the corresponding Jacobi algorithms,
the Jacobi sign is updated for each such partial
quotient. Most of the partial quotients are determined from truncated
inputs where the least significant parts of the numbers are ignored.
The least significant two bits, needed for the Jacobi computation, must
therefore be maintained separately.

\subsection{Notation}

The time needed to multiply two $n$-bit numbers is denoted $M(n)$,
where $M(n) = O(n \log n)$ for the fastest known algorithms.
\footnote{Multiplication in \GMP is based on the more practical
Schönhage-Strassen algorithm, with asymptotic complexity
$O(n \log n \log \log n)$.}

The Jacobi symbol is denoted $(a | b)$. We use the convention that
$[\text{condition}]$ means the function that is one when the condition
is true, otherwise 0, e.g., $(0 | b) = [b = 1]$.

\section{Left-to-right GCD}

In this paper, we will not describe the details of fast \GCD
algorithms. Instead we will consider Algorithm~\ref{alg:gcd}, which is
a generic left-to-right \GCD algorithm, with a basic reduction step
where a multiple of the smaller number is subtracted from the larger
number. We also describe the main idea of fast instantiations of this
algorithm.

\begin{algorithm}
  \centering
  \begin{codebox}
    \Procname{$g \gets \GCD(a,b)$}
    \zi In: $a, b > 0$
    \li \Repeat
    \li   \If $a \geq b$ \Then
    \li     $a \gets a - m b$, with $1 \leq m \leq \floor{a/b}$
    \li     \If $a = 0$ \Then
    \li       \Return $b$
            \End
    \li   \Else
    \li     $b \gets b - m a$, with $1 \leq m \leq \floor{b/a}$
    \li     \If $b = 0$ \Then
    \li       \Return $a$
            \End
          \End
        \End
  \end{codebox}
  \caption{Generic left-to-right \GCD algorithm.}
  \label{alg:gcd}
\end{algorithm}

This algorithm terminates after a finite number of steps, since in
each iteration $\max(a,b)$ is reduced, until $a = b$ and the algorithm
terminates. It returns the correct value, since $\GCD(a,b)$ is
unchanged by each reduction step.

The running time of an instantiation of this algorithm depends on the
choice of $m$ in each step, and on the amount of computation done in
each step. E.g., if $m = 1$, the worst case number of iterations in
exponential. Euclid's algorithm is a special case where, in each step,
$m$ is the correct quotient of the current numbers.

The faster algorithms implements an iteration that depends only on
some of the most significant bits of $a$ and $b$: These bits determine
which of $a$ and $b$ is largest, and they also suffice for computing
an $m$ which is close to the quotient $\floor{a/b}$ or $\floor{b/a}$.
Furthermore, one can compute an initial part of the sequence of
reductions based on the most significant parts of $a$ and $b$, collect
the reductions into a transformation matrix, and apply all the
reductions at once to the least significant parts of $a$ and $b$ later
on. This saves a lot of time, since it omits computing all the
intermediate $a$ and $b$ to full precision. If one repeatedly chops
off one or two of the most significant words, one gets Lehmer's
algorithm, and by chopping numbers in half, one can construct a
divide-and-conquer algorithm with subquadratic complexity.

We will extend this generic algorithm to also compute the Jacobi
symbol. To do that, we need to investigate how the basic reduction $a
- m b$ affects the Jacobi symbol. When we have sorted this out, in the
next section, the result is easily applied to all variants of
Algorithm~\ref{alg:gcd}.

\section{Left-to-right Jacobi}

In this section, we summarize the properties of the Jacobi symbol we
use, derive the update rules needed for our left-to-right algorithm.
Finally, we give the resulting algorithm and prove its correctness.

\subsection{Jacobi symbol properties}

The Jacobi symbol $(a | b)$ is defined for $b$ odd and positive, and
arbitrary $a$. We work primarily with non-negative $a$, and make use
of the following properties of the Jacobi symbol.
\begin{proposition}
  Assume that $a$ is positive and that $b$ is odd and positive. Then
  \begin{enumerate}[(i)]
  \item \label{it:zero}%
 $(0 | b) = [b = 1]$.
  \item \label{it:negation}%
    $(a | b) = (-1)^{(b-1)/2} (-a | b)$
  \item \label{it:reciprocity}%
    If both $a$ and $b$ are odd, then
    \begin{equation*}
      (a | b) = (-1)^{(a-1)(b-1)/4} (b | a)      
    \end{equation*}
  \item \label{it:odd-reduction} %
    $(a | b) = (a - m b | b)$ for any $m$.

  \item \label{it:even-reduction-4}%
    If $a = 0 \pmod 4$ and $1 \leq m \leq \floor{b/a}$, then
    \begin{equation*}
      (a | b) = (a | b - ma)      
    \end{equation*}
  \item \label{it:even-reduction-2}%
    If $a = 2 \pmod 4$ and $1 \leq m \leq \floor{b/a}$, then
    \begin{equation*}
      (a | b) = (-1)^{m(b-1)/2 + m(m-1)/2} (a | b - ma)      
    \end{equation*}
  \end{enumerate}
\end{proposition}
\begin{proof}
  For~\eqref{it:zero} to~\eqref{it:odd-reduction} we refer to standard
  textbooks. The final two are not so well-known, and their use for
  Jacobi computation is suggested by
  Schönhage~\cite{schoenhage-brent-communication}. To
  prove them, assume that $a$ is even and $a < b$. Then
  \begin{align*}
    (a | b) &= (a - b | b) && \text{By \eqref{it:odd-reduction}}\\
    &= (-1)^{(b-1)/2} (b - a | b) && \text{By \eqref{it:negation}} \\
    &= (-1)^{(b-1)/2 + (b-1)(b-a-1)/4} (b | b-a) &&\text{By \eqref{it:reciprocity}}\\
    &= (-1)^{(b-1)/2 + (b-1)(b-a-1)/4}(a | b-a) && \text{By \eqref{it:odd-reduction}}
  \end{align*}
  Since $b^2 - 1$ is divisible by 8 for any odd $b$, we get a
  resulting exponent, modulo two, of
  \begin{equation*}
    (b-1)/2 + (b-1)(b-a-1)/4 = a (b-1)/4 
  \end{equation*}
  If $a = 0 \pmod 4$, this exponent is even and hence there is no sign
  change. And this continues to hold if the subtraction is repeated,
  which proves~\eqref{it:even-reduction-4}. Next, consider the case $a
  = 2 \pmod 4$. Then $a/2 = 1 \pmod 2$, and repeating the subtraction
  $m$ times gives the exponent
  \begin{multline*}
    a \{(b-1) + (b-a-1) + \cdots + (b - (m-1) a - 1)\} / 4 \\
    = m(b-1)/2 + m(m-1)/2 \pmod 2
  \end{multline*}
  which proves~\eqref{it:even-reduction-2}. 
\end{proof}

Finally, note that in these formulas, all the signs are determined by
the least significant two bits of $a$, $b$ and $m$.

\subsection{The new algorithm}

The \GCD algorithm works with two non-negative integers $a$ and $b$,
where multiples of the smaller one is subtracted from the larger. To
compute the Jacobi symbol we maintain these additional state
variables:
\begin{align*}
e &\in \Z_2 && \text{Current sign is $(-1)^e$} \\
\alpha &\in \Z_4 && \text{Least significant bits of $a$} \\
\beta &\in \Z_4 && \text{Least significant bits of $b$} \\
d &\in \Z_2 && \text{Index of denominator}
\end{align*}

The value of $d$ is one if the most recent reduction subtracted $b$
from $a$, and zero if it subtracted $a$ from $b$. We collect these
four variables as the state $S = (e, \alpha, \beta, d)$. The state is
updated by the function \proc{jupdate}, Algorithm~\ref{alg:jupdate}.
\begin{algorithm}
  \centering
  \begin{codebox}
    \Procname{$S' \gets \proc{jupdate} (S, d', m)$}
    \zi In: $d' \in \Z_2$, $m \in \Z_4$, $S = (e, \alpha, \beta, d)$
    \li \If $d \neq d'$ and both $\alpha$ and $\beta$ are odd \Then
    \li   $e \gets e + (\alpha - 1)(\beta - 1)/4$ \RComment Reciprocity
        \End
    \li $d \gets d'$
    \li \If $d = 1$ \Then
    \li   \If $\beta = 2$ \Then
    \li     $e \gets e + m (\alpha - 1)/2 + m (m-1)/2$
          \End
    \li   $\alpha \gets \alpha - m \beta$
    \li \Else
    \li   \If $\alpha = 2$ \Then
    \li     $e \gets e + m (\beta - 1)/2 + m (m-1)/2$
          \End
    \li   $\beta \gets \beta - m \alpha$
        \End
    \li \Return $S' = (e, \alpha, \beta, d)$
  \end{codebox}  
  \caption{Updating the state of the Jacobi symbol computation.}
  \label{alg:jupdate}
\end{algorithm}
Since the inputs of this function are nine bits, and the outputs are
six bits, it's clear it can be implemented using a lookup table
consisting of $2^9$ six-bit entries, which fits in 512 bytes if
entries are padded to byte boundaries.\footnote{One quarter of the
  entries in this table corresponds to invalid inputs, since at least
  one of $\alpha$ and $\beta$ is always odd. If we also note that the
  value of $d$ is needed only when $\alpha = \beta = 3$, the state can
  be encoded into only 26 values, and then the table can be compacted
  to only 208 entries.}

Algorithm~\ref{alg:jacobi} extends the generic left-to-right \GCD
algorithm to compute the Jacobi symbol. The main loop of this
algorithm differs from Algorithm~\ref{alg:gcd} only by the calls to
\proc{jupdate} for each reduction step.
\begin{algorithm}
  \centering
  \begin{codebox}
    \Procname{$j \gets \proc{jacobi} (a, b)$}
    \zi In: $a, b > 0$, $b$ odd
    \zi Out: The Jacobi symbol $(a | b)$
    \zi State: $S = (e, \alpha, \beta, d)$ 
    \li \label{li:jacobi-init}%
        $S \gets (0, a \bmod 4, b \bmod 4, 1)$ 
    \li \Repeat
    \li   \If $a \geq b$ \Then
    \li \label{li:update-a}%
            $a \gets a - m b$, with $1 \leq m \leq \floor{a/b}$
    \li \label{li:jacobi-update-a}%
            $S \gets \proc{jupdate}(S, 1, m \bmod 4)$
    \li     \If $a = 0$ \Then
    \li       \Return  $[b = 1] (-1)^e$
            \End
    \li   \Else
    \li     $b \gets b - m a$, with $1 \leq m \leq \floor{b/a}$
    \li \label{li:jacobi-update-b}%
            $S \gets \proc{jupdate}(S, 0, m \bmod 4)$
    \li     \If $b = 0$ \Then
    \li       \Return $[a = 1] (-1)^e$
            \End
          \End
        \End
  \end{codebox}
  \caption{The algorithm for computing the Jacobi symbol.}
  \label{alg:jacobi}
\end{algorithm}

\subsection{Correctness}

Let $a_0$ and $b_0$ denote the original inputs to
Algorithm~\ref{alg:jacobi}. Since the reduction steps and the stop
condition are the same as in Algorithm~\ref{alg:gcd}, it terminates
after a finite number of steps. We now prove that it returns $(a_0 | b_0)$.

Algorithm~\ref{alg:jacobi} clearly maintains $\alpha = a \bmod 4$ and
$\beta = b \bmod 4$. We next prove that the following holds at the
start of each iteration:

If $d = 0$ we have
\begin{equation}
  (a_0 | b_0) = (-1)^e \times
  \begin{cases}
    (b | a) & \text{$\alpha$ odd} \\
    (a | b) & \text{$\alpha$ even}
  \end{cases}
\end{equation}
and if $d = 1$ we have
\begin{equation}
  \label{eq:invariant-1}
  (a_0 | b_0) = (-1)^e \times
  \begin{cases}
    (a | b) & \text{$\beta$ odd} \\
    (b | a) & \text{$\beta$ even}
  \end{cases}
\end{equation}
This clearly holds at the start of the loop, to prove that it is
maintained, consider the case $a \geq b$ (the case $a < b$ is
analogous). Let $a$, $b$ (unchanged) and $S = (e, \alpha, \beta, d)$
denote the values of the variables before line~\ref{li:update-a}.
There are a couple of different cases, depending on the state:
\begin{itemize}
\item If $\beta$ is odd and either $\alpha$ is even or $d = 1$, then $(a_0 |
  b_0) = (-1)^e (a | b) = (-1)^e (a - m b | b)$.
\item If $\alpha$ and $\beta$ are both odd and $d = 0$, then $(a_0 | b_0) =
  (-1)^e (b | a) = (-1)^{e + (a-1)(b-1)/4} (a - m b | b)$.
\item If $\beta = 0 \pmod 4$, then $(a_0 | b_0) = (-1)^e (b |
  a) = (-1)^e (b | a - m b)$.
\item If $\beta = 2 \pmod 4$, then $(a_0 | b_0) = (-1)^e (b |
  a) = (-1)^{e + m(a-1)/2 + m(m-1)/2} (b | a - m b)$.
\end{itemize}
In each case, the call to \proc{jupdate} makes the appropriate change
to $e$, and Eq.~\eqref{eq:invariant-1} holds after the iteration.

\section{Results}

The algorithm was implemented in \GMP-5.1.0, released 2012. In
benchmarks at the time, comparing the old binary algorithm to the new
Jacobi extension of Lehmer's \GCD algorithm (both $O(n^2)$), the new
algorithm computed Jacobi symbols about twice as fast for moderate
size numbers (around 2000 bits), and 10 times faster for numbers of
size of 500000 bits. For even larger numbers, the Jacobi extension of
subquadratic \GCD brought even greater speedups.

\section*{Acknowledgments}

The author wishes to thank Richard Brent for providing valuable
background material and for encouraging the writing of this paper.

\bibliographystyle{plain}
\bibliography{ref}

\end{document}